\documentclass[
final
]{dmtcs-episciences}

\usepackage[utf8]{inputenc}
\usepackage{subfigure}

\usepackage{amsmath,amssymb,amsthm,graphicx,fixmath,tikz,latexsym,color,xspace,hyperref}
\usepackage{lineno}

\graphicspath{{figures/}}
\def\Re{\ensuremath{\mathbb R}}

\def\eps{{\varepsilon}}
\providecommand{\remove}[1]{}

\newcommand{\emphi}[1]{\emph{#1}}
\newcommand{\VC}{\ensuremath{\mathsf{VC}}\xspace}

\newcommand{\E}{\ensuremath{\mathcal{E}}}

\newcommand{\C}{\mathcal{C}}

\newcommand{\pth}[1]{\!\left({#1}\right)}

\newcommand{\alert}[1]{\textbf{\color{red}
[[[#1]]]}\marginpar{\textbf{\color{red}**}}\typeout{ALERT:
\the\inputlineno: #1}}

\newcommand{\old}[1]{{{}}}

\newcommand{\namedref}[2]{\hyperref[#2]{#1~\ref*{#2}}}

\newtheorem{theorem}{Theorem}
\newtheorem{definition}{Definition}
\newtheorem{corollary}{Corollary}
\newtheorem{lemma}{Lemma}

\usepackage[round]{natbib}

  \author{Nicolas Grelier\affiliationmark{1} \thanks{Research supported by the Swiss National Science Foundation within the collaborative DACH project Arrangements and Drawings as SNSF Project 200021E-171681.}\and
Saeed Gh. Ilchi\affiliationmark{1}\\ \and
Tillmann Miltzow\affiliationmark{2}\thanks{Veni grant EAGER}\and
Shakhar Smorodinsky\affiliationmark{3}\thanks{Grant 635/16 from the Israel Science Foundation.}}

\title{On the VC-dimension of half-spaces with respect to convex sets}
\affiliation{
  Department of Computer Science, ETH Z{\"u}rich, Switzerland\\
 Department of Computer Science, Utrecht University, Netherlands\\
 Department of Mathematics, Ben-Gurion University of the NEGEV, Beer-Sheva, Israel}
  

\received{2020-07-10}
\revised{2021-01-28, 2021-06-28}
\accepted{2021-07-08}
\begin{document}
\publicationdetails{23}{2021}{2}{2}{6631}
\maketitle
\begin{abstract}
A family $S$ of convex sets in the plane defines a hypergraph $H = (S,\E)$ with $S$ as a vertex set and $\E$ as the set of hyperedges as follows.
Every subfamily $S'\subset S$ defines a hyperedge in $\E$ if and only if
there exists a halfspace $h$ that fully contains $S'$, and
no other set of $S$ is fully contained in $h$.
In this case, we say that $h$ realizes $S'$.
We say a set $S$ is shattered, if all its subsets
are realized. The \VC-dimension of a hypergraph $H$
is the size of the largest shattered set. We show that the \VC-dimension for \emph{pairwise disjoint} convex sets in 
the plane is bounded by $3$, and this is tight.
In contrast, we show the \VC-dimension of convex sets in the plane 
(not necessarily disjoint) is unbounded. We provide a quadratic lower bound in the number of pairs of intersecting sets in a shattered family of convex sets in the plane.
We also show that the \VC-dimension is unbounded for pairwise disjoint convex sets in $\Re^d$, 
for $d\geq 3$.
We focus on, possibly intersecting, segments in the plane and determine that the
\VC-dimension is at most $5$. And this is tight, 
as we construct a set of five segments that can be shattered. 
We give two exemplary applications. One for a geometric set
cover problem and one for a range-query data structure problem, to
motivate our findings.
\begin{center}\includegraphics[scale=0.74]{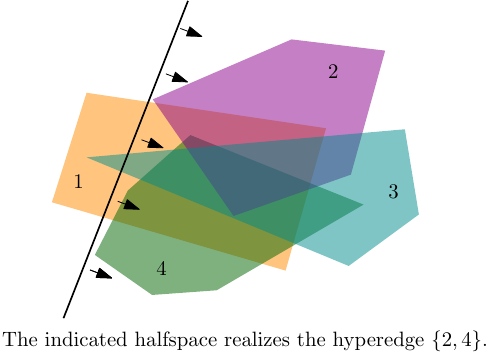}\end{center}
\end{abstract}

\section{Introduction}
Geometric hypergraphs (also called range-spaces) are central objects in computational geometry, statistical learning theory,
combinatorial optimization, linear programming, discrepancy theory,
databases and several other areas in mathematics and computer
science.

In most of these cases, we have a finite set $P$ of points in $\Re^d$ and a family of simple geometric regions, such as say, the family of all halfspaces in $\Re^d$. Then we consider the combinatorial structure of the set system $(P,\{h\cap P\})$
where $h$ is any halfspace. A key property that such hypergraphs have is the so-called bounded \VC-dimension (see later in this section for exact definitions). More precisely, when the underlying family consists of points, the \VC-dimension of the corresponding graph is at most $d+1$. Many optimization problems can be formulated on such structures.
In this paper we initiate the study of a more complicated structure by allowing the underlying set of vertices to be arbitrary convex sets and not just points.
We show that when the underlying family consists of pairwise disjoint convex sets in the plane then
the corresponding hypergraph has \VC-dimension at most $3$ and this is tight. In this case the bound on the \VC-dimension is the same as for points, however we explain later why the proof for pairwise disjoint convex sets has to be more technical. We also show that when the sets may have intersection, then the \VC-dimension is unbounded. Moreover, we prove that even for pairwise disjoint convex sets in $\Re^d$ the \VC-dimension is unbounded already for~$d\geq 3$. This is in sharp contrast to the situation when the underlying family consists of points.

We note that many deep results that hold for arbitrary hypergraphs with bounded \VC-dimension readily apply to such hypergraphs. This includes, e.g., bounds on the discrepancy of such hypergraphs, bounds of $O(\frac{1}{\eps^2})$ on the size of $\eps$-approximations and also bounds on matchings or spanning trees with (so-called) low crossing numbers (see, e.g., \cite{CW89,Mat95,MWW93,LLS01}).

\paragraph*{Preliminaries and Previous Work}\label{sec:prelim}
A hypergraph $H=(V,\E)$ is a pair of sets such that $\E \subseteq
2^V$. A geometric hypergraph is one that can be realized in a
geometric way. For example, consider the hypergraph $H = (V,\E)$,
where $V$ is a finite subset of $\Re^d$ and $\E$ consists of all
subsets of $V$ that can be cut-off from $V$ by intersecting it
with a shape belonging to some family of ``nice'' geometric shapes,
such as the family of all halfspaces. See Figure~\ref{fig:Circle-rangespace},
for an illustration of a hypergraph induced by points in the plane with respect to disks.

\begin{figure}[htbp]
    \centering
    \includegraphics[scale=0.8]{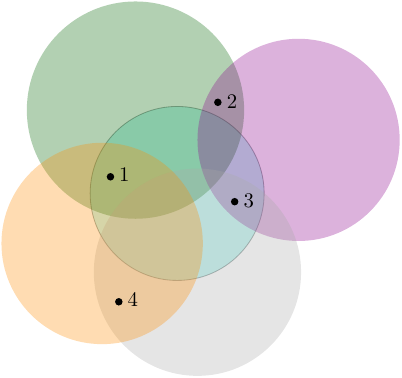}
    \caption{This is a representation of the hypergraph $H = (V, \E)$, with $V= \{1,2,3,4\}$ and for each hyperedge in \E{} of size $2$, a disk realizing the hyperedge is drawn. Note that \E{} contains all subsets of  $V$, except $\{2,4\}$.}
    \label{fig:Circle-rangespace}
\end{figure}

The elements of $V$ are called {\em vertices}, and the elements of $\E$ are called {\em hyperedges}.

We consider the following kinds of geometric hypergraphs:
Let $\C$ be a family of convex sets in $\Re^2$ (or, in general, in $\Re^d$). We say that a subfamily $S \subseteq \C$ is
{\em realized} if there exists a halfspace $h$ such that
$S=\{C \in \C \mid C \subset h \}$. In words, there exists a halfspace $h$ such that the subfamily of $\C$ of all sets that are fully contained in $h$ is exactly $S$.  We refer to the hypergraph $H=(\C,\{ S \mid \textup{S is realized}\})$ as the hypergraph induced by $\C$.
In the literature, hypergraphs
that are induced by points with respect to geometric regions of
some specific kind are also referred to as {\em range spaces}.
We start by introducing the concept of \VC-dimension.

\paragraph*{VC-dimension and $\eps$-nets}
A subset $T \subset V$ is called a \emphi{transversal} (or a \emphi{hitting set}) of a
hypergraph $H=(V,\E)$, if it intersects all sets of $\E$. The
\emphi{transversal number} of $H$, denoted by $\tau(H)$, is the
smallest possible cardinality of a transversal of $H$. The
fundamental notion of a transversal of a hypergraph is central in
many areas of combinatorics and its relatives. In computational
geometry, there is a particular interest in transversals, since
many geometric problems can be rephrased as questions on the
transversal number of certain hypergraphs. An important special
case arises when we are interested in finding a small size set $N
\subset V$ that intersects all ``relatively large'' sets of
$\E$.  This is captured in the notion of an $\eps$-net for a
hypergraph:
\begin{definition}[$\eps$-net]
    Let $H=(V,\E)$ be a hypergraph with $V$ finite. Let $\eps \in
    [0,1]$ be a real number. A set $N \subset V$ (not necessarily in
    $\E$) is called an \emphi{$\eps$-net} for $H$ if for every hyperedge $S
    \in \E$ with $|S| \geq \eps|V|$ we have $S \cap N \neq
    \emptyset$.
\end{definition}

In other words, a set $N$ is an $\eps$-net for a hypergraph $H =
(V,\E)$ if it stabs all ``large'' hyperedges (i.e., those of
cardinality at least $\eps |V|$). The well-known result of~\cite{HW-eps-net} provides a combinatorial condition
on hypergraphs that guarantees the existence of small $\eps$-nets
(see below). This requires the following well-studied notion of the
Vapnik-Chervonenkis dimension \cite{VC71}:

\begin{definition}[\VC-dimension]
Let $H=(V,\E)$ be a hypergraph. A subset $X \subset V$ (not
necessarily in $\E$) is said to be \emphi{shattered} by $H$ if
$\left | \{X\cap S\colon  S \in \E\} \right|=2^X$. The \emphi{Vapnik-Chervonenkis
dimension}, also denoted the \emphi{\VC-dimension} of $H$, is the
maximum size of a subset of $V$ shattered by $H$.
\end{definition}

\paragraph*{Relation between $\eps$-nets and the \VC-dimension}

\cite{HW-eps-net} proved the following fundamental
theorem regarding the existence of small $\eps$-nets for
hypergraphs with small \VC-dimension.

\begin{theorem}[$\eps$-net theorem]\label{theo_epsnet}
    Let us consider $H=(V,\E)$ a hypergraph with \VC-dimension equal to $d$. For
    every $\eps \in (0,1)$, there
    exists an $\eps$-net $N \subset V$ with cardinality at most
    $\displaystyle O\pth{  \frac{d}{\eps}\log\frac{1}{\eps} }$.
\end{theorem}

In fact, it can be shown that a random sample of vertices of size $O(\frac{d}{\eps}\log\frac{1}{\eps})$ is an
$\eps$-net for $H$ with a positive constant probability~\cite{HW-eps-net}.

Many hypergraphs studied in computational geometry and learning theory
have a ``small'' \VC-dimension,
where by ``small'' we mean a constant independent of
the number of vertices of the underlying hypergraph.
In general, range spaces involving semi-algebraic sets of
constant description complexity, i.e., sets defined as a Boolean combination of
a constant number of polynomial equations and inequalities of constant maximum degree,
have finite \VC-dimension. Halfspaces, balls, boxes, etc. are examples of
ranges of this kind; see, e.g.,~\cite{MATOUSEK,PA95} for more details.

Thus, by Theorem~\ref{theo_epsnet}, these hypergraphs admit ``small'' size $\eps$-nets.
\cite{KPW} proved that the bound
$O(\frac{d}{\eps}\log\frac{1}{\eps})$ on the size of an $\eps$-net
for hypergraphs with \VC-dimension $d$ is best possible. Namely,
for a constant $d$, they construct a hypergraph $H$ with
\VC-dimension $d$ such that any $\eps$-net for $H$ must have
size of at least $\Omega(\frac{1}{\eps}\log\frac{1}{\eps})$. Recently, several breakthrough results provided better (lower and upper) bounds on the size of $\eps$-nets in several special cases \cite{Alon-nets,AES09,PachT11}.

In summary, the \VC-dimension is a central notion
in many areas. It proved to be a useful concepts with many
applications. To the best of our knowledge the \VC-dimension
has not been studied for the geometric hypergraphs 
introduced in this paper.

\paragraph*{Results}
We look at a selection of natural geometric hypergraphs 
that arise in our setting.
Our main contribution is to determine its \VC-dimension precisely,
in all cases that we consider.

\begin{theorem} \label{thm:mainConvexIntersecting}
For convex sets in the plane, possibly intersecting, the \VC-dimension is unbounded.
\end{theorem}

\begin{theorem} \label{thm:mainConvexSpace}
For convex disjoint sets in $\Re^d$, for $d\geq 3$, the \VC-dimension is unbounded.    
\end{theorem}

\begin{theorem} \label{thm:mainConvexDisjoint}
For convex disjoint sets in $\Re^2$ the \VC-dimension is at most 3 and this is tight.
\end{theorem}

\begin{theorem} \label{thm:mainSegments}
For segments in the plane, possibly intersecting, the \VC-dimension is at most 5 and this is tight.
\end{theorem}




In order to show the relevance of our findings to the
field of algorithms, we give two simple exemplary 
applications that follow
easily together with previous work.

\paragraph*{Algorithmic Applications}

For our first expository application, we consider a natural
hitting set problem.
\begin{definition}[Hitting Halfplanes with segments]
    Given a set $H$ of halfplanes and a set $S$ of segments, the 
    \emph{halfspace-segment-hitting set problem} asks for a minimum 
    set $T\subset S$,
    such that every halfplane $h\in H$ contains at least
    one segment $t\in T$ entirely.
    This is an optimization problem, where we try to minimize 
    the size of $T$.
\end{definition}

Using the framework of~\cite{bronnimann1995almost}, 
we get the following theorem.
\begin{corollary}
    There is an $O(\log c)$-approximation algorithm
    for the halfspace-segment-hitting set problem,
    where $c$ is the size of the optimal solution.
\end{corollary}

\begin{proof}
We use Theorem~\ref{thm:mainSegments} for segments in the plane and the framework from~\cite{bronnimann1995almost}.
\end{proof}

As a second expository application, we have to introduce the 
problem of approximate range counting. 
Given a family of sets $O$ and a halfplane $h$, we denote by
$n(h,O)$ the number of sets fully contained in $h$.
We denote by $r(h,O) = \frac{n(h,O)}{|O|}$ the relative 
number of sets.
We want to construct a data structure that
reports a number $t$ such that
\[|r(h,O) - t| < \varepsilon,\]
for some given $\varepsilon$.
Thus, here we allow an absolute error rather 
than a relative error.
A simple way is to construct a data structure
is to sample a family of sets $P\subseteq O$ 
and query
how many objects of $P$ are fully contained inside $h$.
If this set $P$ is small we can do queries fast.
\begin{corollary}
    Let  $O$ be a set of disjoint convex objects in the plane.
    Then there
    exists a set $P\subseteq O$ of size
     $O(\frac{1}{\varepsilon^2})$ such that
    $|r(h,O) - r(h,P)| < \varepsilon$.
\end{corollary}
\begin{proof}
    We use Theorem~\ref{thm:mainConvexDisjoint} for disjoint convex sets in the plane and the results from~\cite{LLS01}.
\end{proof}
Note that there are many different notions of approximate
range counting and we presented here a simple one.
Recall that we only want to highlight the relevance of 
our findings for algorithmic applications.

\paragraph*{Structure}
In Section~\ref{sec:non-disjoint},
we show Theorems~\ref{thm:mainConvexIntersecting}
and~\ref{thm:mainConvexSpace}.
In Section~\ref{sec:plane}, we handle the case
of disjoint sets in the plane, which shows
Theorem~\ref{thm:mainConvexDisjoint}.
In Section~\ref{sec:segments}, we show 
Theorem~\ref{thm:mainSegments}.
In Section~\ref{sec:Intersections}, we will consider the minimum number of intersections that shattered families of sets in the plane must have.

\section{Convex sets in the plane and higher dimensions}
\label{sec:non-disjoint}
In this section, we show that when the underlying convex 
sets may intersect, the \VC-dimension can be unbounded.

\begin{figure}[htbp]
    \centering
    \includegraphics[page = 2]{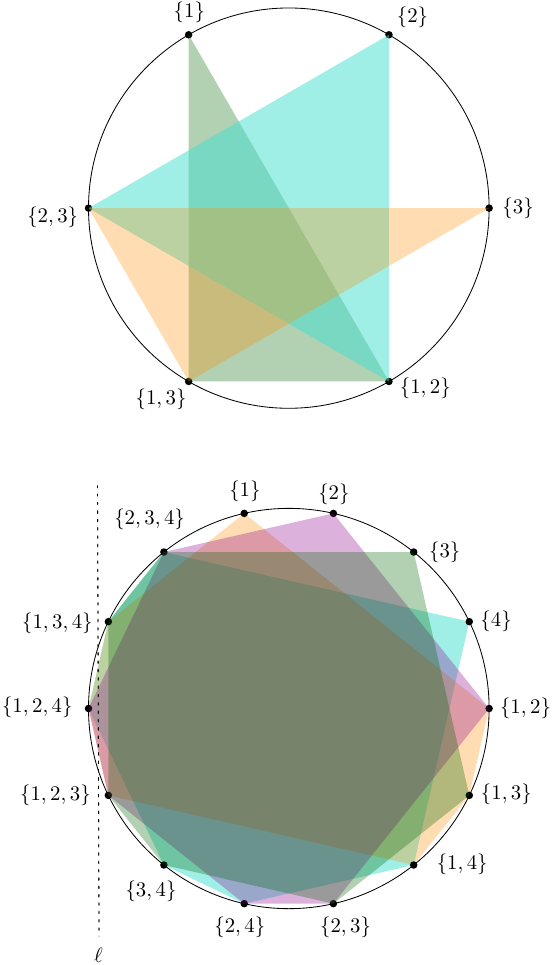}
    \caption{Illustration of the construction of unbounded
    \VC-dimension, for the case $n=3$ (top) and $n = 4$ (bottom). 
    The halfspace to the right of $\ell$ realizes the hyperedge $\{C_3\}$.}
    \label{fig:ConvexInfinitVC}
\end{figure}

\begin{proof}[of Theorem~\ref{thm:mainConvexIntersecting}]
For any $n>0$, we provide $n$ convex sets in the plane that can be shattered.
We denote $[n]=\{1,2,\ldots,n\}$. 
We place $2^n-2$ points on 
the unit circle in the plane as follows.
For every non-trivial subset $I \subset [n], I \neq \emptyset, I \neq [n]$ we place a point $p_I$ on that unit circle.
For each $j \in [n]$ we define the convex set $C_j$ 
as the convex hull
of all points $p_I$ for which $j \in I$. 
Namely $C_j= \textup{conv}(\{p_I \mid j\in I\})$. We claim that the family $\C=\{C_1,\ldots,C_n\}$ is shattered. To see this, let $S \subseteq \C$. If $S$ is either empty or the whole family $\C$ 
then it is easy to see that it is realized as there is a halfplane containing all sets and there is also a halfplane containing none of the sets. So let $I$ be the corresponding non-trivial set of indices corresponding to the members of $S$. Let us denote by $J$ the set $[n]\setminus I$. Consider a line $\ell$ that separates the point $p_J$ from all other points, see Figure~\ref{fig:ConvexInfinitVC} for an illustration. We claim that the halfplane $h$ bounded by $\ell$ and containing those points realizes the subfamily $S$. Indeed notice that for each $C_i \in S$ all points $p_K$ for which $i \in K$ are contained in $h$ so their convex hull $C_i$ is also contained in $h$. Note also that for any $C_j \notin S$ we have that $j \in J$ so $C_j$ contains the point $p_J$ and hence it is not fully contained in $h$. This shows that $S$ is realized for any $S$ and hence $\C$ is shattered.
\end{proof}

\begin{proof}[of Theorem~\ref{thm:mainConvexSpace}]
For any $n>0$, we provide $n$ disjoint convex sets in $\Re^d$ that can be shattered. Let $C_1, \dots, C_n$ be the set of convex shapes that we construct in the proof of Theorem~\ref{thm:mainConvexIntersecting}.
Map each point of $C_i$ such as $(x,y)$ to a point $(x,y,i)$ from $\Re^2$ to $\Re^3$. With this mapping, all the convex sets will be disjoint and we can still shatter these sets as before, by considering vertical halfspaces.
See Figure~\ref{fig:Dimension3} for an illustration. The case for $d>3$
follows in the same way.

\begin{figure}[htbp]
    \centering
    \includegraphics[scale=0.8]{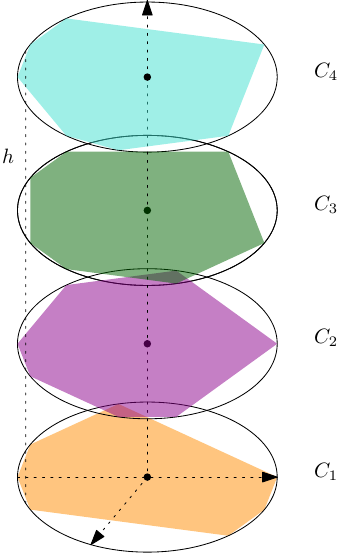}
    \caption{The sets of the two dimensional construction 
    stacked disjointly on top of one another.}
    \label{fig:Dimension3}
\end{figure}

\end{proof}

\section{Disjoint convex sets in the plane}
\label{sec:plane}
From the previous section, we know that the \VC-dimension is unbounded in the plane when shapes can be intersecting. Here we study the case where all the shapes are disjoint.
\begin{figure}
    \centering
    \includegraphics[scale=0.75]{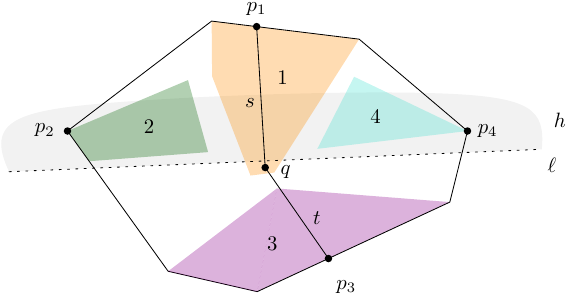}
    \caption{Notation used in the proof of Lemma~\ref{lem:disjoint-convex-plane}.}
    \label{fig:DisjointVC3}
\end{figure}

\begin{lemma}\label{lem:disjoint-convex-plane}
Let $\C$ be a family of pairwise disjoint convex sets in the plane $\Re^2$.  Then, the hypergraph induced by $\C$ has \VC-dimension at most $3$.
\end{lemma}

\begin{figure}
    \centering
    \includegraphics[scale=0.75]{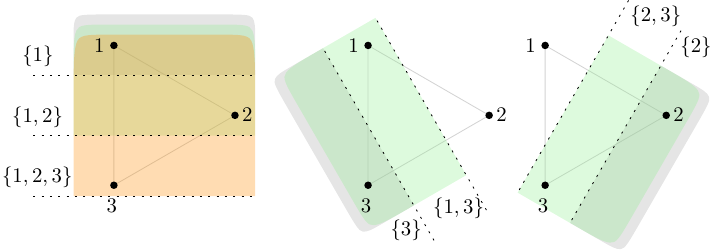}
    \caption{A set of three convex sets that can be shattered. All halfplanes, except the one which contains no convex set, are drawn in the figure.}
    \label{fig:Shatter-three}
\end{figure}

Note that it is easy to find \emph{three} 
disjoint convex sets that can be shattered, see Figure~\ref{fig:Shatter-three}.
As mentioned earlier, a hypergraph induced by a family of points in in the plane $\Re^2$ has \VC-dimension at most $3$ too. The usual proof uses Radon's theorem: Consider four points, they can be divided into two subsets $A$ and $B$ such that the convex hulls of $A$ and $B$ intersect, finally observe that no halfplane can realize $A$ nor $B$. As an illustration, in Figure~\ref{fig:DisjointVC3}, any halfplane that contains the points $p_2$ and $p_4$ must also contain $p_1$ or $p_3$. However, the halfplane denoted by~$h$ in the figure includes the sets $2$ and $4$, but does not include the sets $1$ and~$3$. Therefore $h$ realizes $\{2,4\}$. To show that no family of four pairwise disjoint convex sets is shattered by halfplanes, we need further arguments.
We first prove the following useful lemma.

\begin{lemma}[Convex Hull]\label{lemma:convex hull}
Let $\C$ be a family of sets in the plane. If $\C$ is shattered, then each set in $\C$ contains a point on the boundary of the convex hull of $\C$.
\end{lemma}


Note that in Lemma~\ref{lemma:convex hull} the sets need not to be convex.

\begin{proof}
Let us assume to the contrary that there exists a set $C$ contained in the convex hull of $\C'$, where $\C'$ is a proper subset of $\C$. Then any halfplane containing all elements in $\C'$ must also contain $C$. Therefore, it is not possible to realize $\C'$, which implies that $\C$ is not shattered.
\end{proof}

\begin{proof}[of Lemma~\ref{lem:disjoint-convex-plane}]
Let us assume by contradiction that there exists a shattered family $\C =\{1,2,3,4\}$ of four disjoint convex sets. 
For each convex set $i$ in $\C$, we denote by $p_i$ a point in $i$ that lies on 
the boundary of the convex hull of $\C$. The existence of this point is assured 
by Lemma~\ref{lemma:convex hull}. Without loss of generality, 
let us assume that $p_2$ and $p_4$ have the same $y$-coordinate, 
$p_2$ to the left of $p_4$, with $p_1$ above and $p_3$ below them. 
By assumption, there exist a halfplane $h$ containing $2$ and $4$ 
but not $1$ nor $3$. In particular, $h$ contains $p_2$ and $p_4$. 
However, as $p_1$ and $p_3$ are on the boundary of the convex hull of $\C$, 
$h$ must contain at least one of them, say $p_1$. We denote by $\ell$ the bounding 
line of $h$, which is therefore below the segment between the points $p_2$ and $p_4$. 
As the set $\{2,4\}$ is realized by $h$, $1$ must contain 
a point $q$ below $\ell$. We denote by $s$ the segment between 
the points $p_1$ and $q$. Likewise, we denote by $t$ the segment 
between the points $q$ and $p_3$. Finally, we denote by $r$ 
the union of $s$ and $t$. Let us consider the boundary of the convex hull of $\C$. The points $p_1$ and $p_3$ split it into two curves, one that contains $p_2$ and the other that contains $p_4$. Let us consider the union of the curve that contains $p_2$ with the curve $r$. We have obtained a simple closed curve, which by the Jordan curve theorem splits the plane into two parts. In particular, it splits the convex hull of $\C$ into two parts. We say that the part which contains $p_2$ is to the \emph{left} of $r$, and the other part to the \emph{right} of $r$.
As $1$ is convex, $s$ is fully contained inside
$1$, as its endpoints are contained in $1$. 
By assumption, all $i$ are pairwise disjoint.
Thus $2$ is not intersecting 
with $s$, therefore all points in $2$ are to the left of $r$, as $t$ lies below $\ell$ and $2$ is above $\ell$. 
By the same argument, all points in $4$ are to the right of $r$. 
Note that any halfplane realizing $\{1,3\}$ contains $p_1$, $q$ and $p_3$. By convexity, it contains the triangle with vertices $p_1$, $q$ and $p_3$, and in particular it contains $r$. Thus, it would also contain $2$ or $4$, which is a contradiction. 
\end{proof}

\section{Segments}
\label{sec:segments}
A line segment in the plane can be viewed as the simplest convex set that is not a point.
We now turn to study the special case of the \VC-dimension of hypergraphs induced by line segments.

\begin{lemma}\label{lem:plane-segments}
Let $S$ be a set of (not necessarily disjoint) line segments in $\Re^2$. 
Then the hypergraph induced by $S$ has \VC-dimension at most $5$. 
\end{lemma}

\begin{figure}[htbp]
    \centering
    \includegraphics[page=2,scale=1.1]{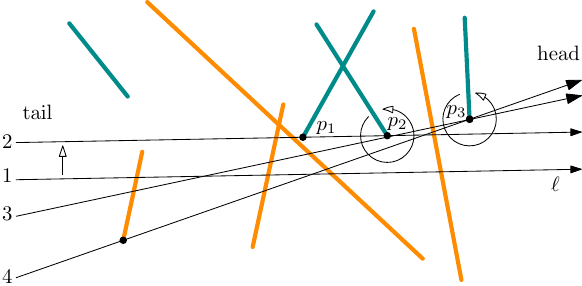}
    \caption{The halfspace $h$ realizing some 
    set $S'$ indicated in turquoise. 
    The line $\ell$ bounds $h$. First, we move $\ell$ up, and then we 
    rotate it around $p_2$ and $p_3$ in that order respectively.}
    \label{fig:Tangent-Construction}
\end{figure}

Before proceeding with the proof we need the following lemma.
We say that a set of segments is \emph{in general position}, if no
three endpoints are collinear. We give an upper bound on the number of subsets that can be realized, by relating this number to the number of tangents to pairs of segments. To the best of our knowledge, we do not know of any previous result that uses the same argument.

\begin{lemma}\label{lem:Tangent-Argument}
Let $S$ be a set of $n$ segments in the plane, in general position. 
Then the number of subsets of $S$ that are realized is at most $2n\left(n-1\right) + 2$.
\end{lemma}
\begin{proof}
Let $h$ be a halfplane realizing a subset $S' \subseteq S$,
with $S'\not = \emptyset$ and $S'\not = S$.  See Figure~\ref{fig:Tangent-Construction} for an illustration.
In the first step, we identify a unique tangent line $\ell$, 
by some transformation argument.
In the second step, we show that every pair of segments has
at most four tangent lines.
Thus, together with the trivial subsets of $S$, we can realize at most \[4 \binom{n}{2} + 2 = 2n(n-1) + 2\] 
subsets $S'$.

\begin{figure}[htbp]
    \centering
    \includegraphics[page=2,scale=0.9]{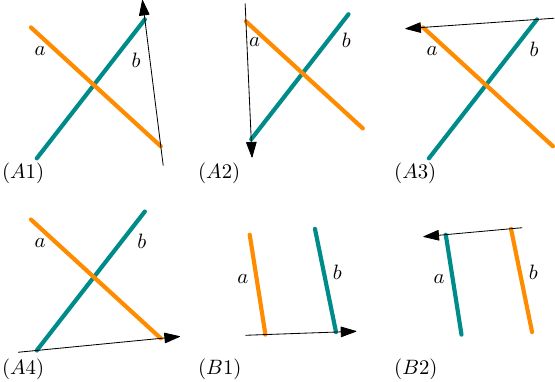}
    \caption{Given two segments there are at most
    four configurations that can arise. All of them 
    are displayed. Note that if the two segments are not crossing
    then only two configurations are possible.}
    \label{fig:Rotation-Configurations}
\end{figure}
 We denote by $\ell$ the bounding line. We orient
 $\ell$ from tail to head such that $S'$ lies to the left of $\ell$.
 If there are several points on $\ell$ then we
 can clearly say, which is closest to its head in the obvious way.
 Translate $h$ inward until its boundary line $\ell$ hits one element of $S'$.
 This must happen as $S'\not = \emptyset$.
 As the set $S$ is in general position,  $\ell$ touches $S'$ 
 in at most two endpoints $p,q$. Suppose that, we touch indeed
 two points, the other case is handled similarly.
 Furthermore, we say that $q$ is the point closer to
 the head of $\ell$.
 Then we rotate $\ell$ counterclockwise around $q$, 
 up until one of two events happen.
 \begin{itemize}
     \item[(a)] The line $\ell$ touches another vertex of some segment $s\in S'$ at its head.
     \item[(b)] The line $\ell$ touches an endpoint of some segment $s \in (S\setminus S')$ at its tail.
 \end{itemize} 
 Note that it could also be that $\ell$ touches 
 a vertex of some segment $s \in (S\setminus S')$ at its head.
 We ignore that case, as this event does not change whether
 $h$ realizes $S'$ or not.
 It is easy to see that it is impossible that 
 $\ell$ touches another vertex of some segment $s\in S'$ at its tail.
 In case (a), we touch a new point $q'$ and we proceed as before.
 In other words, we rotate $\ell$ counterclockwise around $q'$, up until,
 either (a) or~(b) will happen.
 In case~(b), we stop. Note that since $S'\not = S$,
 this will eventually happen.
 We will end up in a configuration,
 where $\ell$ touches a vertex of a segment $s\in S'$ at its head
 and a vertex of another segment $t\in (S\setminus S')$ at its tail.
 Note that both segments are to the left of $\ell$, with respect
 to the orientation of $\ell$.
 Note that the halfspace $h$ defined by $\ell$ only needs an infinitesimally small rotation to realize the original set $S'$
 that we started with. Thus if there were any halfspace realizing $S'$, there must be one of the special type, that we just described.
 This shows the first step. In the second step, we will upper bound 
 the number of those special configurations.

 For the second step, consider two segments $a,b \in S$.
 See Figure~\ref{fig:Rotation-Configurations} for an illustration.
 Note first that they are either crossing or they are disjoint.
 One of them must be contained in the set $S'$ that we want to realize
 and the other is not.
 This also immediately tells us the orientation of the line
 in the configuration. It is easy to check that
 all four configurations are displayed in
 Figure~\ref{fig:Rotation-Configurations}.
\end{proof}

\begin{proof}[of Lemma~\ref{lem:plane-segments}]
Let $S$ be a set of $n$ line segments that can be shattered.
We can assume that $S$ is in general position,
by some standard perturbation arguments.
We will use the fact that the number of distinct subsets of $S$ 
that are realized is at most $2n\left(n-1\right) + 2$, see Lemma~\ref{lem:Tangent-Argument}. 
As there are $2^n$ subsets that need to be realized, for 
$S$ to be shattered, we can conclude that 
it follows that $2^n \leq 2n\left(n-1\right) + 2$. 
However, this inequality is violated for $n\geq 6$, so $n \leq 5$.
\end{proof}

\begin{figure}[htbp]
    \centering
    \includegraphics[page=2,scale=0.8]{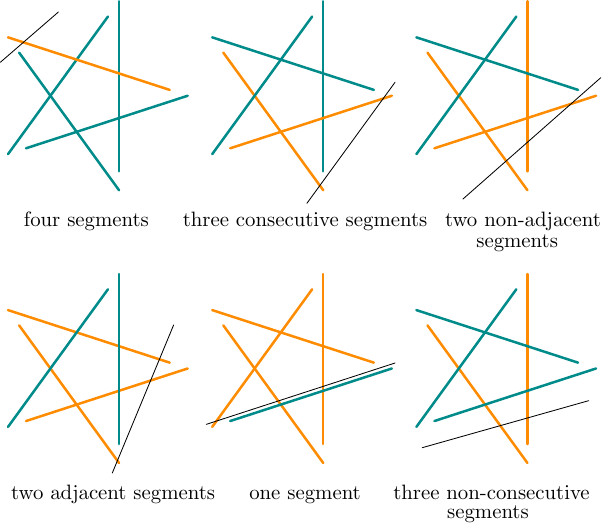}    
    \caption{Shattering five segments as in Lemma~\ref{lemma:five segments}.
    The shattered set is marked turquoise, the other segments are coloured orange.
    Due to the symmetry of the set of segments, we have shown that all hyperedges
    can be realized and thus the set can be shattered.}
    \label{fig:five-Segments}
\end{figure}

The next lemma shows the second part of Theorem~\ref{thm:mainSegments}.
\begin{lemma}\label{lemma:five segments}
There exists a set of five segments that are shattered by halfplanes.
\end{lemma}

\begin{proof}
The set is shown in Figure~\ref{fig:five-Segments}. The segments that are realized are shown in turquoise, and the other in orange.
It is easy to find a halfplane realizing none or all the segments. 
We show in Figure~\ref{fig:five-Segments} how to realize all the remaining configurations.
\end{proof}

\section{Number of intersections}
\label{sec:Intersections}

From Lemma~\ref{lem:disjoint-convex-plane} we have proven that any shattered set of $n$ convex sets are not pairwise disjoint when $n\geq 4$. We show that there are quadratically many pairs of intersecting convex sets.

\begin{lemma}
In a shattered set of $n$ convex sets there are at least $n(n-3)/6$ intersections.
\end{lemma}

\begin{proof}
Consider the intersection graph $G$ of the convex sets. As for any four vertices there is an edge, we obtain that the independence number of $G$ is at most $3$.
(The independence number of a graph denotes the size of the largest independent set of the graph.)
Therefore there is no $K_4$ in the complement of $G$. Tur{\'a}n's theorem states that any graph with $n$ vertices not containing $K_{k+1}$ has at most $(1-1/k)\cdot n^2/2$ edges~\cite{turan1941external}. Therefore, there are at most $(1-1/3)\cdot n^2/2=n^2/3$ non-edges in $G$. 
This is equivalent to having at least 
$\binom{n}{2} -n^2/3=n(n-3)/6$ edges in $G$. 
\end{proof}

\begin{figure}[htbp]
    \centering
    \includegraphics[page=2,scale=0.9]{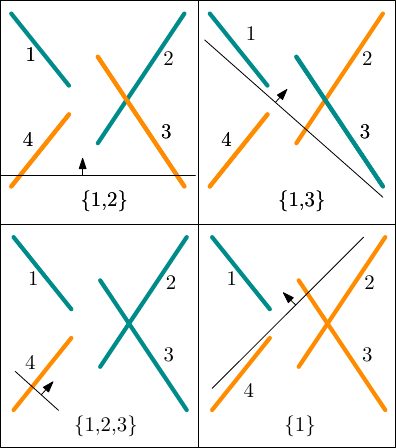}
    \caption{Shattering a set of $4$ segments, with only one intersection.
    We indicate how to shatter the sets $\{1,2\},\{1,3\},\{1,2,3\}$ 
    and~$\{1\}$. 
    All other sets follow the same principle.}
    \label{fig:Four-Segments-Intersection}
\end{figure}

It would be interesting to find an upper bound on how few intersections there may be in a shattered set of $n$ convex sets. The question can also be asked for $n\leq 5$ when considering the more specific case of segments. We have given in Lemma~\ref{lemma:five segments} a shattered set of five segments with five intersections. We produce now an example of a shattered set with four segments having only one intersection.

\begin{lemma}\label{lem:four-segments-intersection}
There exists a shattered set of four segments with only one intersection.
\end{lemma}

\begin{proof}
 We consider the four segments as in Figure~\ref{fig:Four-Segments-Intersection}, 
 denoted by $\{1,2,3,4\}$. 
 It is easy to realize none or all segments. 
 To realize three of them, consider a halfplane whose bounding line intersects the fourth segment. Likewise to realize two consecutive segments, consider a halfplane whose bounding line intersects the two remaining segments.
 For opposite segments, say $\{1,3\}$, take a halfplane not containing $4$ whose bounding line intersects $2$. To realize $\{2,4\}$, take a halfplane not containing $1$ whose bounding line is parallel to $4$ and intersects $3$. Finally the reader can check that for each set with exactly one segment~$i$, it is possible to find a halfplane containing only~$i$.
\end{proof}

\section{Open questions}
As mentioned in Section~\ref{sec:Intersections}, it would be interesting to find tighter lower bounds on the number of intersections in a shattered set of $n$ convex sets in the plane. Likewise, we can ask the same question when the convex sets are constrained to be segments. By Lemma~\ref{lem:disjoint-convex-plane}, we know that for any shattered set of at least four convex sets, there are two convex sets intersecting. We have shown in Lemma~\ref{lem:four-segments-intersection} that it is possible to find a shattered set of four convex sets, with only one intersection. Even more, this holds under the additional constraint that the convex sets be segments. Therefore, we ask whether this holds for any $n$: Is the lower bound on the number of intersections the same whether we consider segments or general convex sets? If not, what is the lower bound when considering polygons with $k$ vertices?

By Theorem~\ref{thm:mainConvexIntersecting}, the VC-dimension of convex sets in the plane is unbounded. However, when restricting to segment, we have shown in Theorem~\ref{thm:mainSegments} that the VC-dimension is at most $5$, and this is tight. The problem of finding upper bounds on the VC-dimension naturally generalizes to other types of constrained convex sets, for instance polygons with $k$ vertices.

\acknowledgements
This work was initiated during the 17th Gremo Workshop on Open Problems 2019. The authors would like to thank the other participants for interesting discussions during the workshop.
We are also grateful to the organizers for the invitation
and a productive and pleasant working atmosphere.

\nocite{*}
\bibliographystyle{abbrvnat}
\bibliography{references}


\end{document}